\documentclass[11pt]{article}

\def\showauthornotes{0}

\usepackage{amsfonts}
\usepackage{amssymb}
\usepackage{amstext}
\usepackage{amsthm}
\usepackage{bm}
\usepackage{color}
\usepackage{comment} 
\usepackage{enumitem}
\usepackage{fancybox}
\usepackage{hyperref}
\usepackage{ifthen}
\usepackage{subfigure}
\usepackage{tabularx}
\usepackage{textcomp}
\usepackage{thm-restate}
\usepackage{thmtools}
\usepackage{url}
\usepackage{verbatim}
\usepackage{xcolor}
\usepackage{mathtools}

\newcommand{\defeq}{\stackrel{\textup{def}}{=}}

\newtheorem{theorem}{Theorem}[section]

\newtheorem{lemma}[theorem]{Lemma}

\newtheorem{fact}[theorem]{Fact}

\theoremstyle{definition}
\newtheorem{definition}[theorem]{Definition}

\DeclareMathOperator*{\pr}{\mathbb{P}}
\DeclareMathOperator*{\av}{\mathbb{E}}
\newcommand{\ceil}[1]{\left\lceil\, {#1}\,\right\rceil}
\newcommand\rea{\mathbb R}
\newcommand{\ol}[1]{\ensuremath{\overline{#1}}}
 
\newcommand\calL{\mathbf{\mathcal{L}}}
\definecolor{Mygray}{gray}{0.8}

 \ifcsname ifcommentflag\endcsname\else
  \expandafter\let\csname ifcommentflag\expandafter\endcsname
                  \csname iffalse\endcsname
\fi

\ifnum\showauthornotes=1
\newcommand{\todo}[1]{\colorbox{Mygray}{\color{red}#1}}
\else
\newcommand{\todo}[1]{}
\fi

\ifnum\showauthornotes=1
\newcommand{\Authornote}[2]{{\sf\small\color{red}{[#1: #2]}}}
\newcommand{\Authoredit}[2]{{\sf\small\color{red}{[#1]}\color{blue}{#2}}}
\newcommand{\Authorcomment}[2]{{\sf \small\color{gray}{[#1: #2]}}}
\newcommand{\Authorfnote}[2]{\footnote{\color{red}{#1: #2}}}
\newcommand{\Authorfixme}[1]{\Authornote{#1}{\textbf{??}}}
\newcommand{\Authormarginmark}[1]{\marginpar{\textcolor{red}{\fbox{%
#1:!}}}}
\else
\newcommand{\Authornote}[2]{}
\newcommand{\Authoredit}[2]{}
\newcommand{\Authorcomment}[2]{}
\newcommand{\Authorfnote}[2]{}
\newcommand{\Authorfixme}[1]{}
\newcommand{\Authormarginmark}[1]{}
\fi

\newcommand{\paren}[1]{\left({#1}\right)}
\newcommand{\sqparen}[1]{\left[{#1}\right]}
\newcommand{\curlyparen}[1]{\left\{{#1}\right\}}
\def\pleq{\preccurlyeq}
\def\pgeq{\succcurlyeq}

\def\ple{\prec}

\def\norm#1{\left\| #1 \right\|}

\newcommand\PPi{\boldsymbol{\Pi}}

\newcommand\ttau{\boldsymbol{\tau}}

\newcommand\bb{\boldsymbol{\mathit{b}}}

\newcommand\vv{\boldsymbol{\mathit{v}}}
\newcommand\ww{\boldsymbol{\mathit{w}}}
\newcommand\yy{\boldsymbol{\mathit{y}}}

\newcommand\xx{\boldsymbol{\mathit{x}}}

\renewcommand\AA{\boldsymbol{\mathit{A}}}
\newcommand\BB{\boldsymbol{\mathit{B}}}
\newcommand\CC{\boldsymbol{\mathit{C}}}
\newcommand\DD{\boldsymbol{\mathit{D}}}

\newcommand\II{\boldsymbol{\mathit{I}}}

\newcommand\MM{\boldsymbol{\mathit{M}}}
\newcommand\LL{\boldsymbol{\mathit{L}}}

\renewcommand\SS{\boldsymbol{\mathit{S}}}
\newcommand\UU{\boldsymbol{\mathit{U}}}
\newcommand\WW{\boldsymbol{\mathit{W}}}

\newcommand\XX{\boldsymbol{\mathit{X}}}
\newcommand\YY{\boldsymbol{\mathit{Y}}}
\newcommand\ZZ{\boldsymbol{\mathit{Z}}}

\newcommand{\vone}{\boldsymbol{\mathbf{1}}}

\usepackage[backend=biber, isbn=false, style=alphabetic, backref=true, doi=false, url=false, maxcitenames=10, mincitenames=5, maxalphanames=10, maxbibnames=10, minbibnames=5, minalphanames=5, defernumbers=true, sortlocale=en_US]{biblatex}

\usepackage{fullpage}

\usepackage{authblk}

\usepackage[linesnumbered,noend,boxed,noline]{algorithm2e}
\DontPrintSemicolon
\IncMargin{1pt}
\SetFuncSty{textsc}

\SetKwProg{Pro}{Procedure}{}{}
    \SetKwFunction{BC}{BlockCholesky}
    \SetKwFunction{DS}{5DDSubset}
    \SetKwFunction{AC}{ApplyCholesky}
    \SetKwFunction{JC}{Jacobi}
    \SetKwFunction{TW}{TerminalWalks}
    \SetKwFunction{PR}{PreconRichardson}
    \SetKwFunction{AS}{ApproxSchur}

\DeclareMathOperator{\Sc}{\mathsf{SC}}

\newcommand\Lhat{\widehat{\boldsymbol{\mathit{L}}}}
\newcommand\Ltil{\widetilde{\boldsymbol{\mathit{L}}}}

\newcommand\xtil{\boldsymbol{\widetilde{\mathit{x}}}}

\makeatletter
\newcommand\multiref[1]{\@first@ref#1,@}
\def\@throw@dot#1.#2@{#1}%
\def\@set@refname#1{%
    \edef\@tmp{\getrefbykeydefault{#1}{anchor}{}}%
    \xdef\@tmp{\expandafter\@throw@dot\@tmp.@}%
    \ltx@IfUndefined{\@tmp autorefnameplural}%
         {\def\@refname{\@nameuse{\@tmp autorefname}s}}%
         {\def\@refname{\@nameuse{\@tmp autorefnameplural}}}%
}
\def\@first@ref#1,#2{%
  \ifx#2@\autoref{#1}\let\@nextref\@gobble%
  \else%
    \@set@refname{#1}%
    \@refname~\ref{#1}%
    \let\@nextref\@next@ref%
  \fi%
  \@nextref#2%
}
\def\@next@ref#1,#2{%
   \ifx#2@ and~\ref{#1}\let\@nextref\@gobble%
   \else, \ref{#1}%
   \fi%
   \@nextref#2%
}

\makeatother

\begin{document}

\title{A Simple and Efficient Parallel Laplacian Solver}

\author{Sushant Sachdeva}
\author{Yibin Zhao\thanks{This research is supported by an Natural Sciences and Engineering
    Research Council of Canada (NSERC) Discovery Grant RGPIN-2018-06398 and an
    Ontario Early Researcher Award (ERA) ER21-16-284. } }
\affil{
    University of Toronto \\
    \texttt{\{sachdeva,ybzhao\}@cs.toronto.edu}
}

\maketitle

\begin{abstract}
  A symmetric matrix is called a Laplacian if it has
  nonpositive off-diagonal entries and zero row sums. Since the
  seminal work of \citeauthor{SpielmanT04} (\citeyear{SpielmanT04}) on solving
  Laplacian linear systems in nearly linear time, several algorithms
  have been designed for the task.  Yet, the work of \citeauthor{KyngS16}
  (\citeyear{KyngS16}) remains the simplest and most practical sequential
  solver.  They presented a solver purely based on random sampling and
  without graph-theoretic constructions such as low-stretch trees
  and sparsifiers.

  In this work, we extend the result of \citeauthor{KyngS16} to a simple
  parallel Laplacian solver with $O(m \log^3 n \log\log n)$ or 
  $O((m + n\log^5 n)\log n \log\log n)$
  work and
  $O(\log^2 n \log\log n)$ depth using the ideas of block Cholesky
  factorization from \citeauthor{KyngLPSS16} (\citeyear{KyngLPSS16}). 
  Compared to the best known parallel
  Laplacian solvers that achieve polylogarithmic depth due to
  \citeauthor{LPS15} (\citeyear{LPS15}), our solver achieves both
  better depth and, for dense graphs, better work.
\end{abstract}

\section{Introduction}
A weighted undirected graph $G=(V,E,\ww)$ is associated with a
canonical symmetric matrix called the graph Laplacian, $\LL_G \in \rea^{V
  \times V},$ such that $\LL_{ij} = -\ww_{ij}$ and $\LL_{ii} = \sum_{e: e
\owns i} \ww_e.$
Solving systems of Laplacian linear equations $\LL_G\xx = \bb$ is a
fundamental linear algebra problem that arises frequently in
scientific computing~\cite{Strang86, BomanHV04}, semi-supervised
learning on graphs~\cite{ZhuGL03, ZhouBLWS04, BelkinMN04}, or
solving maximum-flow or minimum-cost flow using Interior Point
methods~\cite{DaitchS08, ChristianoKMST11, Madry13, LeeS14}.

The groundbreaking work of \citeauthor{SpielmanT04}~\cite{SpielmanT04} 
gave the first nearly-linear time algorithms for solving Laplacian
linear systems. Since then, there has been considerable work on
faster/better algorithms or algorithms in different computation
models~\cite{KoutisMP10, KoutisMP11, KelnerOSZ13, PS14,
  CohenKMPPRX14,  LPS15, KyngLPSS16, KyngS16, FGL+20, JambulapatiS21}.
Most of these algorithms rely on involved graph-theoretic
constructions such as low-stretch spanning trees~\cite{SpielmanT04,
  KoutisMP10, KoutisMP10, KelnerOSZ13}, graph
sparsification~\cite{SpielmanT04, PS14, CohenKMPPRX14, LPS15, KyngLPSS16,
  JambulapatiS21} and explicit expander graphs~\cite{KyngLPSS16}.

In this work, we propose a new parallel algorithm for solving Laplacian linear
systems based on block Cholesky factorization
(a symmetric block matrix version of Gaussian elimination, see \autoref{sec:pre}
and \eqref{eq:bcf})
.
Similar to~\cite{LPS15,KyngLPSS16}, to build such approximate block Cholesky
factorization, our algorithm eliminates large sub-matrices of the Laplacian
corresponding to ``almost-independent sets'' in the graph.
The challenge in \cite{LPS15} is eliminating these vertices while
generating a sparse approximation to the resulting graph (known as the Schur
complement).
We propose a very simple algorithm to compute sparse approximations to Schur
complements based on short random walks. This generalizes the sampling
from~\cite{KyngS16} and does not increase the edge count in our graph; hence
allowing us to bypass graph sparsification.
Putting these together, we obtain a parallel algorithm based purely
on random sampling that does not invoke any graph theoretic
construction.
In addition, we present an application of our technique for generating a sparse
approximation to a Schur complement. %

\paragraph{Related works}
As mentioned above, most Laplacian solvers reply on graph-theoretic
constructions.
A notable exception is the work of \citeauthor{KyngS16}~\cite{KyngS16} which
gives a very simple nearly-linear time Laplacian solver based on
randomized Gaussian elimination (or Cholesky factorization).
Their algorithm is based purely on random sampling and does not
require any graph theoretic constructions. The algorithm can be
described simply as follows: Eliminate the vertices in a uniformly
random order. At each step of elimination, instead of adding a
complete clique on the neighbors as dictated by Gaussian elimination,
sample a few edges and scale them up appropriately.
Despite its simplicity, their algorithm is inherently sequential and does not
lend itself to a good parallel algorithm.

In the realm of parallel Laplacian solvers,
\citeauthor{PS14}~\cite{PS14} obtained the first parallel algorithm with
almost-linear work and polylogarithmic depth, but their algorithm
requires an involved graph sparsification routine
from~\cite{SpielmanT04, OrecchiaV11} that is both complicated and
requires unspecified large polynomial depth.
\citeauthor{LPS15}~\cite{LPS15} gave a different solver that could achieve
better depth and work. 
It was based on the aforementioned block Cholesky factorization approach, but
it still required a parallel graph sparsification routine and explicit expander
graphs.

In contrast, our algorithm for generating sparse approximations to Schur
complements is based on random walks.
There are deep connections between %
random walks, electrical networks, and spectral graph theory.
Most notably, there is a long line of work that utilizes random walks and
approximations to Schur complements to sample random spanning trees \cite{Bro89,
Ald90, Wil96, KM09, MST14, DPPR17, DKP+17, Sch18}.
Recent work by \citeauthor{DGGP19}~\cite{DGGP19}
uses random walks for dynamically maintaining Schur complements. 
This then lead to a few algorithmic advancements in solving maxflow and mincost
flow problems \cite{GLP22,vdBG+22}.

\paragraph{Our Results}
\sloppy
We present a very simple parallel algorithm for solving Laplacian
systems. If $G$ has $n$ vertices and $m$ edges, our algorithm achieves
$O(\log^2 n \log \log n)$ depth and 
$O(m \log^3 n \log\log n)$ or
$O((m+n\log^5 n) \log n \log \log n)$ total work.
This improves on the
$O(\log^6 n \log^4 \log n)$ depth and, for dense graphs, 
$O(m \log^2 n + n \log^4 n \log\log n)$ work achieved by \cite{LPS15} using
parallel sparsifiers from \cite{KyngPPS17} and the $O(m \log^3 n)$ running time
achieved by \cite{KyngS16}.

Our main results are formally presented in the following two theorems:
\begin{theorem} \label{thm:main}
  There is an algorithm that, given a Laplacian matrix
  $\LL \in \rea^{n \times n}$ with $m$ nonzeros, a vector $\bb \in
  \rea^n,$ and $0 < \epsilon < 1/2$,
  returns in $O(m \log^3 n \log\log n \log 1/\epsilon)$ work and
  $O(\log^2 n \log\log n \log 1/\epsilon)$ depth a vector $\xtil$ such
  that
  $ \norm{\xtil - \LL^+\bb}_{\LL} \leq \epsilon \norm{\LL^+\bb}_{\LL}
  $ with high probability.
\end{theorem}

\begin{theorem}
    \label{thm:alter}
    \sloppy
    There is an algorithm that, given a Laplacian matrix
    $\LL \in \rea^{n \times n}$ with $m$ nonzeros, a vector $\bb \in
    \rea^n,$ and $0 < \epsilon < 1/2$,
    returns in 
    $O((m + n \log^5 n)\log n \log\log n \log 1/\epsilon)$ work and 
    $O(\log^2 n \log\log n \log 1/\epsilon)$ depth a
    vector $\xtil$ such that
    $ \norm{\xtil - \LL^+\bb}_{\LL} \leq \epsilon \norm{\LL^+\bb}_{\LL} $
    with high probability.
\end{theorem}

While we present an analysis based on matrix martingales (see
\autoref{sec:concentration}), a simpler analysis
using standard matrix concentration bounds results in only an additional
$\log n$ factor in the total work.

We give our application for generating a sparse approximation to a Schur
complement in \autoref{sec:scappr}.

\begin{paragraph}{Organization}
    \autoref{sec:pre} gives definitions and facts we use throughout the paper.
    In \autoref{sec:bcf}, we give an outline of our algorithm as well as its key
    components.
    \autoref{sec:abcf} gives a proof of our key result
    \autoref{thm:thmsolvenew}.
    We present the proofs to a few important lemmas in
    \multiref{sec:schur,sec:alpha}.
    Our theorem for approximating a specific Schur complement can be found in
    \autoref{sec:scappr}.
    Finally, \autoref{sec:proof} gives all deferred proofs.
\end{paragraph}

\section{Preliminaries} \label{sec:pre}
A symmetric matrix $\AA$ is positive semidefinite (PSD)
(resp. positive definite (PD)) if, for any vector
$\xx$ of compatible dimension, $\xx^\top \AA \xx \geq 0$ (resp. $\xx^\top
\AA \xx > 0$ ).
Let $\AA$ and $\BB$ be two symmetric matrices of the same dimension,
then we write $\BB \pleq \AA$ or $\AA
\pgeq \BB$ if $\AA-\BB$ is PSD.
The ordering given by $\pleq$ is called Loewner partial order.  Let
$\norm{\AA}$ denote the operator norm of a matrix $\AA$.
For symmetric matrices, it is equal to the largest absolute value of
an eigenvalue of $\AA.$
Given a symmetric matrix with eigenvalue decomposition $\AA = \sum_{i}
\lambda_i \vv_i \vv_i^{\top},$ where $\{\vv_i\}_i$ form an orthonormal
basis, the pseudoinverse is defined as $\AA^+ = \sum_{i: \lambda_i
  \neq 0} \frac{1}{\lambda_i}\vv_i \vv_i^{\top}.$

\begin{fact} \label{fact:inver} If $\AA, \BB$ are PSD with the same
  kernel, then $\AA \pgeq \BB$ if and only if $\BB^+ \pgeq \AA^+$.
\end{fact}

\begin{fact} \label{fact:twoside}
    If $\AA \pgeq \BB$ and $\CC$ is any matrix of compatible dimension, then
   $ \CC \AA \CC^\top \pgeq \CC \BB \CC^\top.$
\end{fact}

We say that $\AA$ is an $\epsilon$-approximation of $\BB$, written $\AA
\approx_\epsilon \BB$ if $e^{-\epsilon}\BB \pleq \AA \pleq e^{\epsilon} \BB$. 
This relation is symmetric. Moreover,
if $\AA \approx_\epsilon \BB$ and $\BB \approx_\delta \CC$, then
$\AA \approx_{\epsilon+\delta} \CC$

Let $G = (V,E, \ww)$ be a weighted undirected graph (possibly with
multi-edges) with edge weights $\ww: E \rightarrow \rea_{\ge 0}$.
The Laplacian $\LL_G$ of $G$ is defined as $\LL_G = \DD - \AA$ where
$\DD$ is a diagonal matrix with $\DD_{uu} = \sum_{e \ni u} \ww(e)$ the
weighted degree of vertex $u$, and $\AA$ a symmetric nonnegative
matrix with 0 diagonal and $\AA_{uv} = \sum_{e : e = (u,v)}\ww(e)$
where $u,v$ are distinct vertices. 
When the context is clear, we write $\LL$ instead of $\LL_G$.
Since different multi-graphs can correspond to the same graph Laplacian, we have
written our algorithms completely with respect to the multi-graphs instead of
matrices.
\begin{fact}
    If $G$ is connected, then the kernel of $\LL$ is the span of vector $\vone$. 
\end{fact}
\noindent We assume for the rest of the paper that all graphs are connected
undirected multi-graphs with non-negative edge weights.

We say that $\xtil$ is an $\epsilon$-approximate solution to the Laplacian
linear system $\LL\xx = \bb$ if 
$
    \norm{\xtil - \LL^+\bb}_{\LL} \leq \epsilon \norm{\LL^+\bb}_{\LL}
$
where $\norm{\xx}_{\LL} = \sqrt{\xx^\top\LL\xx}$.
\medskip

\noindent
\begin{paragraph}{Block Cholesky Factorization and Schur Complement}
Given a bipartition of the rows $F \sqcup C$, the block Cholesky factorization of
a Laplacian $\LL$ is
\begin{equation} \label{eq:bcf}
    \LL = 
    \begin{pmatrix}
        \II & 0 \\
        \LL_{CF}\LL_{FF}^{-1} & \II
    \end{pmatrix}
    \begin{pmatrix}
        \LL_{FF} & 0 \\
        0 & \LL_{CC} - \LL_{CF}\LL_{FF}^{-1}\LL_{FC}
    \end{pmatrix}
    \begin{pmatrix}
        \II & \LL_{FF}^{-1}\LL_{FC} \\
        0 & \II
    \end{pmatrix},
\end{equation}
where the matrix in the lower-right block of the middle matrix is
defined as
the Schur complement of $\LL$ onto $C$
\[
    \Sc(\LL,C) \defeq \LL_{CC} - \LL_{CF}\LL_{FF}^{-1}\LL_{FC}.
\]
\begin{fact}
    If $G$ is connected, then the matrix $\Sc(\LL_G,C)$ is also a Laplacian
    matrix of a connected graph.
\end{fact}

\begin{fact} \label{fact:block}
    If $\LL_{FF} \pleq \Ltil_{FF}$, then
    \[
        \begin{pmatrix}
            \LL_{FF} & \LL_{FC} \\
            \LL_{CF} & \LL_{CC}
        \end{pmatrix}
        \pleq 
        \begin{pmatrix}
            \Ltil_{FF} & \LL_{FC} \\
            \LL_{CF} & \LL_{CC}
        \end{pmatrix}.
    \]
\end{fact}
\end{paragraph}

\begin{paragraph}{Parallel Model and Primitives}
As adopted by all previous works \cite{PS14,LPS15,KyngLPSS16} on parallel
Laplacian solvers, the runtime bounds of this work are in terms of \emph{work}
and \emph{depth} (span).
The work of an algorithm is the total number of operations performed, and the
depth is the length of the longest chain of sequentially dependent operations.
We remark that our algorithms work in the classic CREW PRAM model. 

We use several existing parallel primitives implicitly in our algorithm.
Since our algorithm uses random walks, we frequently need to perform weighted
random samplings and transform between different representations of graphs.
Weighted random sampling asks for sampling items from a set with replacement
such that the probability of sampling an item is proportional to a given
weight.

\begin{lemma}[\cite{HS19}]
    There is a weighted random sampling algorithm that takes $O(n)$ work and
    $O(\log n)$ depth preprocessing time and $O(1)$ work and depth for each
    query.
\end{lemma}
\begin{lemma}[\cite{BM10}]
    There is an algorithm that converts between edge list and adjacency list
    representation of any multi-graph with $m$ multi-edges in $O(m)$ work and
    $O(\log m)$ depth.
\end{lemma}

\end{paragraph}

\section{Block Cholesky Factorization Algorithm} \label{sec:bcf}
In this section, we present our algorithm for solving Laplacian linear
systems, and the key subroutines. 
The algorithm \BC for recursively applying a sequence of block
Cholesky factorization and building the necessary matrices for solving
the system is summarized in \autoref{alg:build}. 
\autoref{alg:solve} summarizes the \AC algorithm for solving Laplacian linear
systems using the matrices generated by \BC to constant approximation.

\subsection{Overview}
The block Cholesky factorization of a Laplacian
$\LL$ from \eqref{eq:bcf} allows us to write $\LL^+$ as
\begin{equation} \label{eq:linv}
    \LL^+ =
    \begin{pmatrix} \II & -\LL_{FF}^{-1}\LL_{FC} \\ 0 & \II
    \end{pmatrix}
    \begin{pmatrix} \LL_{FF}^{-1} & 0 \\ 0 & \Sc(\LL, \CC)^+
    \end{pmatrix}
    \begin{pmatrix} \II & 0 \\ -\LL_{CF}\LL_{FF}^{-1} & \II
    \end{pmatrix}.
\end{equation}
This factorization presents a natural approach to solving Laplacian
systems.
The upper and lower triangular block matrices are easy to
apply, given that the system of $\LL_{FF}$ is easy to solve.
As for the middle block diagonal matrix, since a Schur complement of a Laplacian
matrix is also Laplacian, we can recursively apply the block Cholesky
factorization until the Schur complement is small.
To ensure good parallel runtime, we restrict the number of recursions to 
$O(\log n)$.
Note, however, that the graph corresponding to the Schur complement of a sparse
graph Laplacian can be dense.

Now, our goal is to find a large subset $F$ such that two conditions
hold: 1) $\LL_{FF}$ is an easy matrix to invert, and 2) a sparse
approximation to the Schur complement $\Sc(\LL, V \setminus F)$ can be
computed efficiently.
If both steps can be done efficiently,
then we can recurse on the approximate Schur complement.

We show that it suffices to find $F$ such that $\LL_{FF}$ represents an
``almost independent'' graph. To be precise, we will require $\LL_{FF}$
to be a $5$-DD matrix, that we define as follows.
\begin{definition}[5-DD Matrix]
  A symmetric matrix $\MM$ is said to be $5$-diagonally
  dominant (5-DD) if for each row $i$, $\MM_{ii} \geq 5 \sum_{j:j \neq i}
  |\MM_{ij}|.$
\end{definition}
In the context of a graph, a subset $F \subseteq V$ is said to be $5$-DD if
$\LL_{FF}$ is a $5$-DD matrix.
\DS (\autoref{alg:sdd}) is a subroutine from~\cite{LPS15} that allows us to
find 5-DD subsets efficiently (See \autoref{subsec:sdd}).

Secondly,
we show in \autoref{subsec:trw} that when $F \subseteq V$ is a 5-DD subset, we
can generate a sparse approximation to $\Sc(\LL, V\setminus F)$ using a simple
scheme that samples short random walks (\TW in \autoref{alg:trw}).

Before we present the aforementioned steps, we first introduce the notion of
$\alpha$-boundedness in \autoref{subsec:alpha}. 
This notion guides us in studying
the central object, \emph{leverage score}, in the Laplacian paradigm.

\begin{algorithm}[t]
    \caption{Block Cholesky Factorization Algorithm}
    \label{alg:build}
    \Pro{\BC{$G$}}{
        Initially, let $G^{(0)} \gets G$, $k \gets 0$.\;
        \While{$G^{(k)}$ has more than 100 vertices}{ \label{alg:build_d}
            $k \gets k+1$.\;
            $F_k \gets$ \DS{$G^{(k-1)}$} and set $C_k \gets C_{k-1} \backslash
            F_k$.\;
            $G^{(k)} \gets$ \TW{$G^{(k-1)}, C_k$}.\;
        }
        Let $d$ be the last $k$.\;
        \Return{$(G^{(0)}, \ldots G^{(d)}; F_1, \ldots F_d)$.}\;
    }
\end{algorithm}

\begin{algorithm}[t]
    \caption{Apply Block Cholesky Factorization}
    \label{alg:solve}
    \Pro{\AC{$G^{(0)}, \ldots G^{(d)}; F_1, \ldots F_d; \bb$}}{
        Let $\bb^{(0)} \gets \bb$.\;
        \For{$k \gets 0, \ldots ,d-1$}{
            Apply forward substitution to solve
            $
                \begin{pmatrix}
                    (\ZZ^{(k)})^{-1} & 0 \\
                    (\LL_{G^{(k)}})_{C_{k+1}F_{k+1}} & \II
                \end{pmatrix}
                \begin{pmatrix}
                    \yy^{(k)}_{F_{k+1}} \\
                    \yy^{(k)}_{C_{k+1}}
                \end{pmatrix}
                = 
                \bb^{(k)},
            $ %
            where $\ZZ^{(k)}$ is the equivalent operator of
            \JC{$G^{(k)},F_{k+1}, 1/2d, \cdot$}.\;
            $\bb^{(k+1)} \gets \yy^{(k)}_{C_{k+1}}$.\;
        }
        Solve $\LL_{G^{(d)}}\xx^{(d)} = \bb^{(d)}$.\;\label{alg:solve_d}
        \For{$k \gets d-1, \ldots ,0$}{
            Apply backward substitution to solve
            $
                \begin{pmatrix}
                    \II & \ZZ^{(k)}(\LL_{G^{(k)}})_{F_{k+1}C_{k+1}} \\
                    0 & \II
                \end{pmatrix}
                \begin{pmatrix}
                    \xx^{(k)}_{F_{k+1}} \\
                    \xx^{(k)}_{C_{k+1}}
                \end{pmatrix}
                =
                \begin{pmatrix}
                    \yy^{(k)}_{F_{k+1}} \\
                    \xx^{(k+1)}
                \end{pmatrix}.
            $ \;
        }
        \Return{$\xx^{(0)}$.}\;
    }
    \Pro{\JC{$G, F, \epsilon, \bb$}}{
        Set $(\LL_{G})_{FF} = \XX + \YY$ where $\XX$ is diagonal and
        $\YY$ is the Laplacian of the induced subgraph $G[F]$.\;
        Set $l$ the smallest odd integer such that $l \geq
        \log_2(3/\epsilon)$.\;
        \For{$i = 1 \ldots l$}{
            $\xx^{(i)} \gets \XX^{-1}\bb - \XX^{-1}\YY\xx^{(i-1)}$.\;
        }
        \Return{$\xx^{(l)}$.}\;
    }
\end{algorithm}

\subsection{$\alpha$-boundedness for Multi-edges} \label{subsec:alpha}
For a positive scalar $\alpha$, we say that a multi-edge $e$ is
$\alpha$-bounded w.r.t.~some Laplacian matrix $\LL$ if its leverage score is at
most $\alpha$, i.e. the weight of the multi-edge times its effective resistance
in $\LL$: 
\[
    \ttau(e) = \ww(e) \bb_e^\top \LL^{+}\bb_e \leq \alpha.
\]
We say that a weighted multi-graph $G$ is $\alpha$-bounded w.r.t.~a Laplacian
$\LL$ if every multi-edge of $G$ is $\alpha$-bounded w.r.t.~$\LL$.

For \BC, our input graph $G$ is required to be a connected
undirected multi-graph $G$ such that each multi-edge is $\alpha$-bounded
w.r.t.~$\LL_G$ for some $\alpha < 1$.
Having $\alpha$-bounded leverage scores is essential for good matrix
concentration guarantees (see \autoref{thm:freedman}).
Our algorithm maintains the $\alpha$-boundedness of all newly sampled
multi-edges (see \TW and \autoref{sec:schur}).

Given an input simple graph, there are two ways to achieve
$\alpha$-boundedness initially.
Since $w(e) \bb_e^\top \LL_G^+ \bb_e \leq 1$ for any edge $e$, it suffices to
split each edge into $\ceil{\alpha}$ copies with $1/\ceil{\alpha}$ times the
original weight. 
This approach is used in \autoref{thm:main}. 
\begin{lemma} \label{lemma:split}
    \sloppy
    There is an algorithm that,
    given a weighted simple graph $G$ with $n$ vertices,
    $m$ edges and a positive scalar $\alpha$ satisfying $\alpha^{-1} =
    \textnormal{poly}(n)$,
    returns a multi-graph $H$ with $O(m\alpha^{-1})$
    multi-edges such that each multi-edge of $H$ is $\alpha$-bounded
    w.r.t.~$\LL_H$ and $\LL_H = \LL_G$.
    The algorithm runs in $O(m\alpha^{-1})$ work and $O(\log n)$
    depth. 
\end{lemma}

For certain applications where the graph is relatively dense,
one can improve its work dependency on edges by using a leverage score
overestimation scheme and perform the splitting by such estimations. 
By invoking our Laplacian solver in \autoref{thm:main}, we get
the following guarantees due to \cite{CohenLMMPS15,kyng2017approximate}.
We prove the parallel runtime in \autoref{sec:alpha}.

\begin{lemma} 
    \label{lemma:sparsebound}
    \sloppy
    There is an algorithm that,
    given a weighted simple graph $G$ with $n$ vertices and $m$ edges and a
    positive scalar $\alpha, K$ satisfying $\alpha^{-1} = O(n)$ and
    $0 < K < m$,
    with high probability returns a multi-graph $H$
    with $O(m + nK\alpha^{-1})$ multi-edges such that each multi-edge of
    $H$ is $\alpha$-bounded w.r.t.~$\LL_H$ and $\LL_H = \LL_G$. 
    The algorithm runs in 
    $O(nK\alpha^{-1} + m\log n + \frac{m}{K}\log^4 n \log\log n)$ work and
    $(\log^2 n \log\log n)$ depth.
\end{lemma}

\subsection{Vertex Elimination using Diagonally Dominant Sets} \label{subsec:sdd}

\begin{algorithm}[t]
    \caption{Routine for generating a $5$-DD subset
    \cite{LPS15, KyngLPSS16}}
    \label{alg:sdd}
    \Pro{\DS{$G = (V,E,\ww)$}}{
        Set of vertices $F \gets \emptyset$ initially.\;
        \While{$|F| \leq n/40$}{
            Let $F'$ be a uniformly sampled subset of $V$ of size
            $\frac{n}{20}$.\;
            Let $F$ be the set of all $i \in F'$ such that 
            $\sum_{e \in E(G[F]) : e \ni i}
            \ww(e) \leq \frac{1}{5} \sum_{e \in E: e \ni i} \ww(e)$.\;
        }
        \Return{$F$.}\;
    }
\end{algorithm}

The work of \cite{LPS15} gives an algorithm that finds a $5$-DD subset
of constant fraction size of the vertices with linear work and logarithmic
expected depth.
This is achieved by repeatedly sampling a large subset of vertices and
removing vertices that violate the $5$-DD condition until the
resulting subset is large enough, described in \textsc{5DDSubset}.
\begin{lemma}[\cite{LPS15} Lemma 5.2 paraphrased] 
    \label{lemma:sdd}
    For every multi-graph $G$ with $n$ vertices, $m$ multi-edges,
    \DS computes a $5$-DD subset $F$ of size at least $n/40$ in
    $O(m)$ expected work and $O(\log m)$ expected depth where $m$ is the number
    of multi-edges in $G$.
\end{lemma}

\cite{LPS15} also observed that 5-DD matrices can be solved
efficiently using Jacobi method. 

\begin{lemma}[\cite{LPS15} Lemma 5.4] 
    \label{lemma:jacobi}
    Given a $5$-DD matrix $\MM = \XX + \YY$ where $\XX$ is diagonal,
    $\YY$ is a Laplacian, and a scalar $0 < \epsilon < 1$, the
    operator $\ZZ$ defined as
    \begin{equation} \label{eq:jacobi}
        \ZZ = \sum_{i=0}^l \XX^{-1} (-\YY\XX^{-1})^i,
    \end{equation}
    where $l$ is an odd
    integer such that $l \geq \log_2 3/\epsilon,$ satisfies
    $
        \MM \pleq \ZZ^{-1} \pleq \MM + \epsilon \YY
   $
    and $\ZZ$ can be applied to a vector in work $O(m \log 1/\epsilon)$ and
    depth $O(\log m \log 1/\epsilon)$ where
    $m$ satisfies that $\MM$ can be applied in $O(m)$ work and $O(\log m)$ depth. 
\end{lemma}
The \JC algorithm (\autoref{alg:solve}) implements the
operator described by \autoref{lemma:jacobi}.
Such an approximation can then be used to replace $\LL_{FF}^{-1}$
subblock in the block Cholesky factorization form to generate an
approximation to the Laplacian.
\begin{lemma}[\cite{LPS15} Lemma 5.6 paraphrased] 
    \label{lemma:jacobiapprox}
    Let $\LL$ be a Laplacian matrix, $F$ be a $5$-DD subset with the complement
    subset $C$ and a scalar $0 < \epsilon \leq 1/2$.
    For $\ZZ$ as defined in \eqref{eq:jacobi}, 
    \[
        \begin{pmatrix}
            \II & 0 \\
            \LL_{CF}\ZZ & \II
        \end{pmatrix}
        \begin{pmatrix}
            \ZZ^{-1} & 0 \\
            0 & \Sc(\LL, C)
        \end{pmatrix}
        \begin{pmatrix}
            \II & \ZZ\LL_{FC} \\
            0 & \II
        \end{pmatrix}
        \approx_{\epsilon} \LL.
    \]
\end{lemma}
\noindent For completeness, we have included the proofs of
\multiref{lemma:sdd,lemma:jacobi,lemma:jacobiapprox} in the
\autoref{sec:proof}.

\subsection{C-Terminal Random Walks for Approximating Schur Complements}
\label{subsec:trw}
Formally, given a subset of vertices $C$, we say a walk
$W = (u_0, \ldots, u_l)$ in $G$ is $C$-terminal iff
$u_0, u_l \in C$ and $u_1, \ldots, u_{l-1} \in V \backslash C$.  It is
well known that the Schur complement $\Sc(\LL, C)$ is the sum of all
$C$-terminal walks (\cite{DPPR17} Lemma 5.4) weighted appropriately.
The following lemma extends this result to multi-graphs,
To distinguish the multi-edges used by the walk, we instead write 
$W = (u_0, e_1, u_1, \ldots, u_{l-1}, e_l, u_l)$ where $e_i \ni u_{i-1}, u_i$. 
Once again, we say that $W$ is $C$-terminal iff $W$ is a walk in
multi-graph $G$ and only $u_0,u_l$ are in $C$. 

\begin{algorithm}[t]
    \caption{$C$-Terminal Random Walks}
    \label{alg:trw}
    \Pro{\TW{$G=(V,E,w),C$}}{
        Let $H \gets$ an empty multi-graph\;
        \For{each multi-edge $e \in E$}{
            Let $u,v$ be the two incident vertices of $e$.\;
            Generate a random walk $W_1(e)$ from $u$ until it reaches
            $C$ at some vertex $c_1$.\;
            Generate a random walk $W_2(e)$ from $v$ until it reaches
            $C$ at some vertex $c_2$.\;
            Connect $W_1(e), e, W_2(e)$ to form a walk $W(e)$ between
            $c_1$ and $c_2$ if $c_1 \neq c_2$.\;\label{alg:trw_j}
            Add a multi-edge $f_e = \{c_1, c_2\}$ with weight 
            $\frac{1}{\sum_{f \in W(e)} 1/\ww(f)}$ and let $\YY_e$ be the
            associated Laplacian of this multi-edge.\;\label{alg:trw_w}
        }
        \Return{The multi-graph $H$.}\;
    }
\end{algorithm}

\begin{lemma} 
    \label{lemma:sumofwalks}
    For any multi-graph $G$ with its associated
    Laplacian $\LL$ and any non-trivial subset of vertices 
    $C \subseteq V$, the Schur complement $\Sc(\LL, C)$ is given as an union
    over all multi-edges corresponding to $C$-terminal walks 
    $W = (u_0, e_1, u_1 \ldots, e_l, u_l)$ and with weight
    \begin{equation}
        \frac{\prod_{i=1}^l \ww(e_i)}{\prod_{i=1}^{l-1} \ww(u_i)} 
    \end{equation}
    where $\ww(u) = \sum_{e : e \ni u} \ww(e)$ for any vertex $u$. 
\end{lemma}
We defer the proof to \autoref{sec:proof}.

Our algorithm \textsc{TerminalWalks} samples $C$-terminal walks so that the
resulting multi-graph has its associated Laplacian matrix in expectation the
same as the Schur complement while keeping the total number of samples small.
Since each sample is relatively small, we can use matrix
concentration inequality to prove that the approximation is close. 
In \autoref{sec:concentration} we formally prove these claims.

\subsection{A Nearly Log Squared Depth Solver}
By preconditioned iterative methods for solving Linear systems, so long as we
generate a constant good preconditioner to Laplacian $\LL$ (i.e.~an operator
that is a constant approximation to $\LL^+$ and fast to apply), we can use such
preconditioner to solve linear systems in $\LL$ with only $O(\log 1/\epsilon)$
overhead for precision $\epsilon$. 
It then suffices to build just a constant approximate inverse to $\LL$. 
\begin{theorem}[Preconditioned Richardson Iteration, \cite{S03,A96}]
    \label{thm:richardson}
    There exists an algorithm \textsc{PreconRichardson} such that for
    any matrices $\AA, \BB \pgeq \mathbf{0}$ such that
    $\BB \approx_{\delta} \AA^+$ and any error tolerance
    $0 < \epsilon < 1/2$,
    \textsc{PreconRichardson}($\AA,\BB,\bb,\delta, \epsilon$) returns
    an $\epsilon$-approximate solution $\xtil$ to $\AA\xx = \bb$ in
    $O(e^{2\delta} \log(1/\epsilon))$ iterations, each consisting of
    one multiplication of a vector by $\AA$ and one by $\BB$.
\end{theorem}

\begin{algorithm}[t]
    \caption{Preconditioned Richardson Iteration}
    \label{alg:richardson}
    \Pro{\PR{$\AA, \BB, \bb, \delta, \epsilon$}}{
        Set $\alpha \gets 2/(e^{-\delta}+e^{\delta})$ where
        scalar $\delta > 0$ satisfies $\BB \approx_{\delta} \AA^+$.\;
        Initially set $\xx^{(0)} \gets \BB \bb$.\;
        \For{$k \gets 1 \ldots d=\ceil{e^{2\delta} \log(1/\epsilon)}$}{
            $\xx^{(k)} \gets (\II - \alpha\BB\AA)\xx^{(k-1)} +
            \alpha\xx^{(0)}$.\;
        }
        \Return{$\xx^{(d)}$.}\;
    }
\end{algorithm}

\autoref{thm:bcfapproxnew} states that we can efficiently generate a sparse
block Cholesky factorization that is a constant approximation.
Before we state the theorem, we need to define a series of matrices.
For any positive integer $k$, let $\DD^{(k)}$ be a block diagonal matrix defined
as
\begin{equation} \label{eq:blockd}
    \begin{split}
    \DD^{(k)}
    =&
    \DD^{(k-1)} - \LL_{G^{(k-1)}} + 
    \begin{pmatrix}
        (\LL_{G^{(k-1)}})_{F_kF_k} & 0 \\
        0 & \LL_{G^{(k)}}
    \end{pmatrix}
    \\
    =&
    \begin{pmatrix}
        \DD^{(k-1)}_{V\backslash C_{k-1}} & 0 & 0\\
        0 & (\LL_{G^{(k-1)}})_{F_kF_k} & 0 \\
        0 & 0 & \LL_{G^{(k)}}
    \end{pmatrix}
    \end{split}
\end{equation}
with $\DD^{(0)} = \LL$, 
and $\UU^{(k)}$ be an upper triangular matrix defined by
\begin{equation} \label{eq:blocku}
\begin{split}
    \UU^{(k)} 
    &= 
    \begin{pmatrix}
        \II & 0 & 0 \\
        0 & \II & ((\LL_{G^{(k-1)}})_{F_kF_k})^{-1}(\LL_{G^{(k-1)}})_{F_kC_k} \\
        0 & 0 & \II
    \end{pmatrix}
    \UU^{(k-1)}
    \\
    &=
    \begin{pmatrix}
        \II & \UU^{(k-1)}_{(V\backslash C_{k-1}) F_k} 
        & \UU^{(k-1)}_{(V\backslash C_{k-1}) C_k} \\
        0 & \II & ( (\LL_{G^{(k-1)}})_{F_kF_k})^{-1}(\LL_{G^{(k-1)}})_{F_kC_k} \\
        0 & 0 & \II
    \end{pmatrix}
\end{split}
\end{equation}
with $\UU^{(0)} = \II$ initially.
Note that $\UU^{(k-1)}_{C_kC_k} = \II$.

\begin{theorem} \label{thm:bcfapproxnew}
    \sloppy
    There exists some $\alpha^{-1} = \Theta(\log^2 n)$ such that,
    given any multi-graph $G$ with $n$ vertices,
    $m$ multi-edges such that each multi-edge is $\alpha$-bounded w.r.t.$\LL_G$,
    the algorithm \BC{$G$}
    returns with high probability
    a sequence of multi-graphs and subsets of vertices
    $(G^{(0)}, \ldots G^{(d)}; F_1, \ldots F_d)$ such that 
    \begin{enumerate}[ref=\autoref{thm:bcfapproxnew}-(\arabic*)]
        \item For $0 \leq k < d$, $G^{(k)}$ has at most $m$ multi-edges.
            \label{thm:bcf1}
        \item For $0 < k \leq d$, $F_k$ is a $5$-DD subset of $\LL_{G^{(k-1)}}$.
            \label{thm:bcf2}
        \item $G^{(d)}$ has size $O(1)$.
            \label{thm:bcf3}
        \item $d = O(\log n)$.
            \label{thm:bcf4}
        \item With $\DD^{(k)}$ and $\UU^{(k)}$ defined as in
        \eqref{eq:blockd} and \eqref{eq:blocku} respectively, with high
        probability 
            \[
                (\UU^{(d)})^\top \DD^{(d)} \UU^{(d)} \approx_{0.5} \LL_G.
            \]
            \label{thm:bcf5}
    \end{enumerate}
    The algorithm takes $O(m\log n)$ work and $O(\log m \log n)$ depth
    with high probability. 
\end{theorem}

In the next section, we prove that the above chain of sparse
approximate Schur complements suffices to build an approximate inverse
for $\LL$.
\begin{theorem} \label{thm:thmsolvenew}
    \sloppy
    Let $G$ be any multi-graph with $n$ vertices,
    $m$ multi-edges.
    There exists some $\alpha^{-1} = \Theta(\log^2 n)$ such that
    if each multi-edge of $G$ is $\alpha$-bounded w.r.t.~$\LL_G$
    and let 
    $(G^{(0)}, \ldots, G^{(d)}; F_1, \ldots, F_d)$ 
    be the output of \BC{$G$}, 
    then, with high probability, the algorithm 
    \AC{$G^{(0)}, \ldots, G^{(d)}; F_1, \ldots, F_d; \bb$}
    corresponds to a linear operator $\WW$ acting on $\bb$ such that
    \[
        \WW^+ \approx_{1} \LL_G
    \]
    and the algorithm runs in $O(m\log n \log\log n)$ work and 
    $O(\log m \log n \log\log n)$ depth with high probability.
\end{theorem}

Our main guarantees for multi-graphs with $\alpha$-bounded
multi-edges follows directly by applying preconditioned Richardson iterations. 
\begin{lemma} \label{lemma:mainnew}
    \sloppy
    There exists an algorithm and a positive scalar $\alpha$ satisfying
    $\alpha^{-1} = \Theta(\log^2 n)$ such that, given input a multi-graph $G$
    with $n$ vertices and $m$ multi-edges that are all $\alpha$-bounded
    w.r.t. its associated Laplacian $\LL$, a vector $\bb \in \rea^n$ and $0 <
    \epsilon < 1/2$,
    returns in $O(m \log n \log\log n \log 1/\epsilon)$ work and
    $O(\log m \log n \log\log n \log 1/\epsilon)$ depth a vector $\xtil$ such
    that
    $\norm{\xtil - \LL^+\bb}_{\LL} \leq \epsilon \norm{\LL^+\bb}_{\LL}$
    with high probability.
\end{lemma}
\begin{proof}
    Building the sequence of matrices is given by \autoref{thm:bcfapproxnew}. 
    By \autoref{thm:thmsolvenew}, there is an operator $\WW$ that can be applied in
    $O(m \log n \log\log n)$ work and $O(\log m \log n \log\log n)$ depth to a
    conformable vector and $\WW^+ \approx_{1} \LL$. 
    Using \autoref{fact:inver}, $\WW \approx_{1} \LL^+$ and that we can use
    \textsc{PreconRichardson}($\LL, \WW, \bb, \epsilon$) to generate an
    $\epsilon$ approximation to the system $\LL\xx = \bb$ by
    \autoref{thm:richardson}. 
    Since the algorithm applies $\WW$ for a total of $O(\log 1/\epsilon)$ times,
    the total work and depth guarantee is $O(m \log n \log\log n \log
    1/\epsilon)$ and $O(\log m \log n \log\log n \log 1/\epsilon)$ respectively. 
\end{proof}

Both \autoref{thm:main} and \autoref{thm:alter} are direct consequences of
\autoref{lemma:mainnew} where the former one uses naive edge splitting
due to \autoref{lemma:split} and the latter one uses \autoref{lemma:sparsebound}
with the choice $K = \Theta(\log^3 n)$.

\section{Solving using Block Cholesky Factorization} \label{sec:abcf}
In this section, we prove \autoref{thm:thmsolvenew} for solving the Laplacian
system using an approximate block Cholesky factorization.

\begin{proof}[Proof of \autoref{thm:thmsolvenew}]
    Notice that each $\xx^{(k)}$ is a linear transformation of $\bb^{(k)}$. 
    Let us define $\WW^{(k)}$ to be the equivalent linear operator of such a
    transformation.
    When $k = d$, by \autoref{alg:solve_d} of \autoref{alg:solve} the operator
    is exactly $\WW^{(d)} = \LL_{G^{(d)}}$.
    For simplicity, we write $\calL^{(k)} = \LL_{G^{(k)}}$~\footnote{Not to be
    confused with $\LL^{(k)}$.}.
    For the rest, one can observe that the operator is
    \begin{align*}
        (\WW^{(k)})^+
        =&
        \begin{pmatrix}
            (\ZZ^{(k)})^{-1} & 0 \\
            \calL^{(k)}_{C_{k+1}F_{k+1}} & \II
        \end{pmatrix}
        \begin{pmatrix}
            \II & 0 \\
            0 & (\WW^{(k+1)})^+
        \end{pmatrix}
        \begin{pmatrix}
            \II & \ZZ^{(k)}\calL^{(k)}_{F_{k+1}C_{k+1}} \\
            0 & \II
        \end{pmatrix}
        \\
        =&
        \begin{pmatrix}
            \II & 0 \\
            \calL^{(k)}_{C_{k+1}F_{k+1}}\ZZ^{(k)} & \II
        \end{pmatrix}
        \begin{pmatrix}
            (\ZZ^{(k)})^{-1} & 0 \\
            0 & (\WW^{(k+1)})^+
        \end{pmatrix}
        \begin{pmatrix}
            \II & \ZZ^{(k)}\calL^{(k)}_{F_{k+1}C_{k+1}} \\
            0 & \II
        \end{pmatrix}.
    \end{align*}

    \sloppy
    Let $\LL^{(d,k)} = (\UU^{(d)}_{C_kC_k})^\top
    \DD^{(d)}_{C_kC_k}\UU^{(d)}_{C_kC_k}$
    with $C_0 = V(G)$.
    We use backward induction to show that for all $k$,
    $
        (\WW^{(k)})^+ \approx_{1/2-k/2d} \LL^{(d,k)}
    $
    The base case is already covered for $k=d$.
    For the induction step, assume the approximation for $k+1$, i.e.
    $
        (\WW^{(k+1)})^+ \approx_{1/2-(k+1)/2d} \LL^{(d,k+1)}.
    $
    By \eqref{eq:blockd} and \eqref{eq:blocku}, 
    \begin{align*}
        \LL^{(d,k)} =&
        \begin{pmatrix}
            \II & 0 \\
            \calL^{(k)}_{F_{k+1}C_{k+1}} (\calL^{(k)}_{F_{k+1}F_{k+1}})^{-1}
            & \II
        \end{pmatrix}
        \begin{pmatrix}
            \calL^{(k)}_{F_{k+1}F_{k+1}} & 0 \\
            0 & \LL^{(d,k+1)}
        \end{pmatrix}
        \begin{pmatrix}
            0 & \II \\
            \II &
            (\calL^{(k)}_{F_{k+1}F_{k+1}})^{-1} \calL^{(k)}_{C_{k+1}F_{k+1}} 
        \end{pmatrix}.
    \end{align*}
    Using our inductive assumption and including $(\ZZ^{(k)})^{-1}$ to the
    topleft block,
    \[
        \begin{pmatrix}
            (\ZZ^{(k)})^{-1} & 0 \\
            0 & (\WW^{(k+1)})^+
        \end{pmatrix}
        \approx_{1/2-\frac{k+1}{2d}}
        \begin{pmatrix}
            (\ZZ^{(k)})^{-1} & 0 \\
            0 & \LL^{(d,k+1)}
        \end{pmatrix}.
    \]
    Let $\MM^{(k)}$ be defined as
    \begin{align*}
        \MM^{(k)} = 
        &
        \begin{pmatrix}
            \II & 0 \\
            \calL^{(k)}_{C_{k+1}F_{k+1}}\ZZ^{(k)} & \II
        \end{pmatrix}
        \begin{pmatrix}
            (\ZZ^{(k)})^{-1} & 0 \\
            0 & \LL^{(d,k+1)}
        \end{pmatrix}
        \begin{pmatrix}
            \II & \ZZ^{(k)}\calL^{(k)}_{F_{k+1}C_{k+1}} \\
            0 & \II
        \end{pmatrix}.
    \end{align*}
    Then, using \autoref{fact:twoside} with $\CC$ being the lower triangular
    matrix,
    \[
        (\WW^{(k)})^+ \approx_{1/2-\frac{k+1}{2d}} \MM^{(k)}.
    \]
    Note that $\Sc(\LL^{(d,k)}, C_{k+1}) = \LL^{(d,k+1)}$
    directly from block Cholesky factorization.
    Given the choice of $\epsilon = 1/2d$, by \autoref{lemma:jacobiapprox}, we
    also have
    $\MM^{(k)} \approx_{1/2d} \LL^{(d)}_{C_kC_k}$,
    and that
    \[
        (\WW^{(k)})^+ \approx_{1/2-\frac{k}{2d}} \LL^{(k)}
    \]
    as required.

    Lastly, consider the work and depth of \AC. 
    For the first for loop, the runtime of each forward substitution is dictated
    by two parts: (a) solving using Jacobi iteration with equivalent operator
    $\ZZ^{(k)}$, and (b) apply a subblock of $\calL^{(k)}$ to a vector. 
    Similarly, each backward substitution step, indexed by $k$, is also dictated
    by the same two operations. 

    We remark that the Laplacian of a multi-graph with $m$ multi-edges can be
    applied to a vector in $O(m)$ work and $O(\log m)$ depth. 
    The $O(\log m)$ depth comes from summing up the value for each entry after
    all the edge weights and vector values are multiplied in parallel. 

    As for all $k < d$, $G^{(k)}$ has $O(m)$ multi-edges by
    \ref{thm:bcf1}.
    Thus, applying a subblock of $\calL^{(k)}$ runs in
    $O(m)$ work and $O(\log m)$ depth. %
    Now, $\calL^{(k)}_{F_{k+1}F_{k+1}}$ is a subblock of $\calL^{(k)}$,
    and again by \ref{thm:bcf2}, is also $5$-DD. 
    So, given the choice of $\epsilon = 1/O(\log n)$,
    \autoref{lemma:jacobi} gives us that $\ZZ^{(k)}$ can be applied in 
    $O(m \log\log n)$ work and $O(\log m \log\log n)$ depth. 

    There is only a small overhead of $O(1)$ for both work and depth in solving
    $\calL^{(d)} \xx^{(d)} = \bb^{(d)}$ due to \ref{thm:bcf3}. 
    Therefore, the total runtime is $O(m \log n \log\log n)$ work and 
    $O(\log m \log n \log\log n)$ depth by $d = O(\log n)$.
\end{proof}

\section{Schur Approximation Guarantees} \label{sec:schur}
\label{sec:concentration}

In this section, we formally present the guarantees of the subroutine
\TW in \autoref{alg:trw} and use them to show
\autoref{thm:bcfapproxnew}. 

Firstly, 
we show that our random sampling procedure \TW for
approximation Schur complements is unbiased. 
\begin{lemma} \label{lemma:unbiased}
    For every multi-graph $G$ and every non-trivial $C \subseteq V(G)$,
    \textsc{TerminalWalks}$(G, C)$ returns a multi-graph $H$ such that 
    $\av \LL_H = \Sc(\LL_G, C)$.
\end{lemma}
\begin{proof}
    For each $e$, recall that $\YY_e$ is the associated Laplacian of newly
    sampled edge $f_e$ (See \autoref{alg:trw_w} of \TW). 
    Note that $\YY_e = 0$ if $f_e$ is not formed due to $c_1 = c_2$ as in 
    \autoref{alg:trw_j} of \TW.
    By \autoref{lemma:sumofwalks}, $\Sc(\LL_G, C)$ is sum of $C$-terminal random
    walks on $G$. 

    Consider a multi-edge $e \in G$ and let the incident vertices be $u,v$. 
    For the sampled walk $W(e)$, let $W_1(e) = (u_0, e_1, \ldots, e_p, u_p)$
    with $u_0 = u$ and $W_2(e) = (v_0, f_1, \ldots, f_q, v_q)$ with 
    $v_0 = v$. 
    The probability of $W(e)$ being sampled is
    \begin{align*}
        \pr(W(e)) 
        &= 
        \pr(W_1(e)) \cdot \pr(W_2(e))
        = 
        \frac{\prod_{i=1}^p \ww(e_i)}{\prod_{i=1}^p \ww(u_i)} 
        \cdot
        \frac{\prod_{i=1}^q \ww(f_i)}{\prod_{i=1}^q \ww(v_i)} 
        \\
        &= 
        \frac{1}{\ww(e)} 
        \frac{\prod_{f \in E(W)} \ww(f)}{\prod_{z \in V(W) \backslash C} \ww(z)},
    \end{align*}
    where $V(W)$ and $E(W)$ are the multisets of vertices and edges in $W$. 
    We set $\prod_{z \in V(W) \backslash C} \ww(z) = 1$ when $V(W) \backslash C =
    \emptyset$ to cover the case where both $u,v \in C$. 
    Then, the expected entry for any $c_1, c_2 \in C, c_1 \neq c_2$ is
    \begin{align*}
        \left( \av \sum_e \YY_e \right)_{c_1c_2}
        =& 
        -\sum_e 
        \sum_{
            \begin{subarray}{c}
                W: C\text{-terminal}, e \in W \\
                c_1,c_2 \in V(W) \cap C
            \end{subarray}
            }
        \frac{1}{\sum_{f \in E(W)} 1/\ww(f)} 
        \cdot \frac{1}{\ww(e)} 
        \cdot
        \frac{\prod_{f \in E(W)} \ww(f)}{\prod_{u \in V(W) \backslash C} \ww(u)}
        \\
        =&
        -\sum_{
            \begin{subarray}{c}
                W: C\text{-terminal} \\
                c_1,c_2 \in V(W) \cap C
            \end{subarray}
            }
        \frac{\prod_{f \in E(W)} \ww(f)}{\prod_{u \in V(W) \backslash C} \ww(u)}
        \cdot 
        \frac{1}{\sum_{f \in W} 1/\ww(f)}
        \cdot
        \left( \sum_{e: e \in W} \frac{1}{\ww(e)} \right)
        \\
        =&
        -\sum_{
            \begin{subarray}{c}
                W: C\text{-terminal} \\
                c_1,c_2 \in V(W) \cap C
            \end{subarray}
            }
        \frac{\prod_{f \in E(W)} \ww(f)}{\prod_{u \in V(W) \backslash C} \ww(u)}
        \\
        =& 
        \left( \Sc(\LL_G,C) \right)_{c_1c_2}.
    \end{align*}
    As each $\YY_e$ is Laplacian, the sum must also be Laplacian, and
    subsequently the expected sum is also Laplacian. 
    Hence,
    $%
        \av \LL_H = \av \sum_e \YY_e = \Sc(\LL_G,C). 
    $%
\end{proof}

The following lemma formally lays out an important guarantee of
\textsc{TerminalWalks} for proving our concentration results,
i.e.~$\alpha$-boundedness is ``closed'' under this operation. 
\begin{lemma}
    \label{lemma:twalksnew}
    For every multi-graph $G$ such that each
    multi-edge of $G$ is $\alpha$-bounded w.r.t.~a Laplacian $\LL$,
    and every non-trivial $C \subseteq V(G)$, 
    \textsc{TerminalWalks}$(G, C)$ returns a multi-graph $H$ with a collection
    of independent multi-edges $f_e$ where each is $\alpha$-bounded
    w.r.t.~$\LL$.
\end{lemma}
To show this, we invoke the property that effective resistance is a distance and
satisfies triangle inequality. 
\begin{lemma}[\cite{KyngS16}, Lemma 5.2] \label{lemma:refftri} 
    For every weighted, connected multi-graph $G$ with its associated
    Laplacian $\LL$, any three vertices $u,v,z \in V(G)$ and the pair vectors
    $\bb_{uv}, \bb_{vz}, \bb_{uz}$,
    then
    \[
        \bb_{uz}^{\top} \LL^+ \bb_{uz} \le  \bb_{uv}^{\top} \LL^+ \bb_{uv}
        +  \bb_{vz}^{\top} \LL^+ \bb_{vz}.
    \]
\end{lemma}

\begin{proof}[Proof of \autoref{lemma:twalksnew}]
    If $f_e = e$, then $\ww(f_e) = \ww(e)$ and its $\alpha$-boundedness
    trivially holds. 
    Otherwise, consider the sampled random walk $W(e)$ and say $W(e) = (u_0,
    e_1, \ldots, e_l, u_l)$ for some $l \geq 2$ and w.l.o.g.~$u_0 = c_1, u_l =
    c_l$. 
    By applying \autoref{lemma:refftri} inductively (and recall that $\ttau$ is
    the leverage score),
    \[
        \bb_{u_0 u_l}^\top \LL^+ \bb_{u_0 u_l} 
        \leq
        \sum_{i=1}^l \bb_{e_i}^\top \LL^+ \bb_{e_i}
        =
        \sum_{i=1}^l \frac{\ttau(\bb_{e_i})}{\ww(e_i)}
        \leq
        \alpha \sum_{i=1}^l \frac{1}{\ww(e_i)}. 
    \]
    Note that by $u_0 = c_1 \neq c_2 = u_l$, we can always segment the walk in a
    way to avoid dealing with self-loops during induction steps. 
    Recall that the weight is set to exactly 
    $\ww(f_e) = \frac{1}{\sum_{i=1}^l \frac{1}{\ww(e_i)}}$ by 
    \autoref{alg:trw_w} of \TW.
    Therefore, $f_e$ is $\alpha$-bounded w.r.t.~$\LL$. 
\end{proof}

We show that the number of multi-edges generated by \TW never increases and that
the algorithm runs efficiently.
\begin{lemma} \label{lemma:twtime}
    For every multi-graph $G$ with $n$ vertices and
    $m$ multi-edges, and every non-trivial $C \subseteq V(G)$,
    \textsc{TerminalWalks}$(G, C)$ returns a multi-graph 
    $H$ with at most $m$ multi-edges.
    If the subset $F = V\backslash C$ is $5$-DD in $\LL_G$, then
    the algorithm runs in $O(m)$ work and $O(\log m)$ depth with high
    probability.
\end{lemma}

\begin{proof}[Proof of \autoref{lemma:twtime}]
    For the number of edges, notice that each multi-edge $e \in E(G)$ can spawn
    at most 1 multi-edge in $H$.
    Thus, the number of multi-edges of $H$ is at most $m$.
    
    Now, we analyze the time complexity of \TW. 
    Due to $F$ being $5$-DD, for any vertex $u \in F$, a
    random walk step samples a vertex in $C$ with probability at least $4/5$. 
    Then, the length of the walk $|W(e)| = O(1)$ in expectation. 
    The longest length of all walks is $\max_e |W(e)| = O(\log m)$
    with high probability. 
    In addition, using the Chernoff bound, the sum of length of all walks is 
    $\sum_e |W(e)| = O(m)$ with high probability.
    Thus, it takes $O(m)$ work and $O(\log m)$ depth to sample all the random
    walks. 
    There is an additional $O(m)$ work and $O(\log m)$ depth
    overhead for the weighted sampling preprocessing time and for coverting the
    new multi-edges back into an adjacency list representation.
    Hence, the algorithm runs in $O(m)$ work and $O(\log m)$ depth. 
\end{proof}

We now present the proof to our main guarantees to the block Cholesky
factorization algorithm in \autoref{thm:bcfapproxnew}. 
For the rest of this section, we use $\LL$ to denote the input Laplacian. 
For any PSD matrix $\SS$ such that $\ker(\LL) \subseteq \ker(\SS)$, we
define
\[
    \ol{\SS} \defeq \LL^{+/2} \SS \LL^{+/2}.
\]

Recall the definition to $\DD^{(k)},\UU^{(k)}$ in
\eqref{eq:blockd},\eqref{eq:blocku}.
Now, by \eqref{eq:bcf} the approximate block Cholesky factorization up to
iteration $k$ is then
\[
\begin{aligned}
    \LL^{(k)} 
    =&
    {\UU^{(k)}}^\top \DD^{(k)} \UU^{(k)}
    \\
    =& 
    {\UU^{(k-1)}}^\top \DD^{(k-1)} \UU^{(k-1)} 
    + \LL_{G^{(k)}} - \Sc(\LL_{G^{(k-1)}},C_k)
    \\
    =&
    \LL^{(k-1)}
    + \LL_{G^{(k)}} - \Sc(\LL_{G^{(k-1)}},C_k)
\end{aligned}
\]
with $\LL^{(0)} = \LL$ initially.
By \autoref{lemma:unbiased}, conditional on the choices of \BC up to iteration
$k-1$ and $F_k$, we have
\[
    \av \LL^{(k)} = 
    \LL^{(k-1)}
    + \av \LL_{G^{(k)}} - \Sc(\LL_{G^{(k-1)}},C_k)
    =
    \LL^{(k-1)}.
\]
This establishes a matrix martingale.
Thus, to get matrix concentration results, we use the following well established
Freedman's inequality for matrices.
\begin{theorem}[Matrix Freedman, \cite{Tropp11,Tropp12}] \label{thm:freedman}
    Let $\YY_0, \YY_1, \ldots$ be a matrix martingale whose values are symmetric
    and $n \times n$ matrices, and let $\XX_1, \XX_2, \ldots$ be the difference
    sequence $\XX_i = \YY_i - \YY_{i-1}$. 
    Assume that the difference sequence is uniformly bounded by 
    $\norm{\XX_k} \leq R$ for all $k$. 
    Define the predicatable quadratic variation process of the martingale,
    \[
        \WW_k \defeq \sum_{j=1}^k \av_{j-1} \sqparen{\XX_j^2}.
    \]
    Then, for all $t > 0$ and $\sigma^2 > 0$, 
    \[
        \pr\sqparen{\exists k : \norm{\YY_k} \geq t~\mbox{and}~\norm{\WW_k} \leq
        \sigma^2} 
        \leq
        n \cdot \exp\paren{-\frac{t^2/2}{\sigma^2 + Rt/3}}.
    \]
\end{theorem}

\begin{proof}[Proof of \autoref{thm:bcfapproxnew}]
    \ref{thm:bcf3} is immediate from \autoref{alg:build_d} of \BC.
    \ref{thm:bcf2} is guaranteed by \autoref{lemma:sdd}.

    Consider \ref{thm:bcf1}.
    By \autoref{lemma:twtime}, for all $1 \leq k \leq d-1$, the number
    of multi-edges in $G^{(k)}$ is at most the number of multi-edges in
    $G^{(k-1)}$. 
    Since initially $G^{(0)}$ has exactly $m$ edges, by induction, for all $k
    \leq d-1, G^{(k)}$ must have $O(m)$ edges as required.

    As for \ref{thm:bcf4}, notice that again by \autoref{lemma:sdd}, 
    $|C_k| \leq (1-\frac{1}{40})C_{k-1}$ for all $k$. 
    Then, $d$ is at most 
    \begin{equation} \label{eq:dbound}
        d \leq \log_{40/39}(n) 
        = O(\log n).
    \end{equation}

    We now consider the approximation result in \ref{thm:bcf5}, 
    i.e.
    \[
        (\UU^{(d)})^\top \DD^{(d)} \UU^{(d)} \approx_{0.5} \LL.
    \]
    We show the following stronger statement using induction:  
    \emph{
        For any $0 \leq t \leq d$, with probability at least 
        $1 - \frac{t}{n^\delta}$, for all $0 \leq k \leq t$ the approximation
        \[
            0.7 \LL \pleq \LL^{(k)} = 
            (\UU^{(k)})^\top \DD^{(k)} \UU^{(k)} \pleq 1.3 \LL
        \]
        holds conditional on the choice of \BC up to iteration $k-1$ and
        the choice of $F_k$.
    }
    Notice that $0.7 \geq e^{-0.5}, 1.3 \leq e^{0.5}$ and 
    $1 - \frac{d}{n^\delta} \geq 1 - \frac{1}{n^{\delta-1}}$. 
    We can choose an appropriate constant $\delta > 1$ for high probability.

    The base case is trivial as $(\UU^{(0)})^\top \DD^{(0)} \UU^{(0)} = \LL$. 
    For the inductive case, we start by assuming that all such inequalities
    hold up to $t-1$. 
    For any $k$, conditional on the choices upto step $k-1$ of
    \BC, we assume some random order $<, \leq$ on the
    multi-edges of $G^{(k-1)}$.
    As the choices of $\YY_e$ and $\YY_f$ in \TW are
    independent for $e \neq f$, such random order does not affect our analysis. 
    To simplify our notation, let $\av_{\YY^{(k)}_e}$ denote the conditional
    expectation on all the choices prior to $\YY^{(k)}_e$. 
    Now, for any $k$, let $\XX^{(k)}_e = \YY^{(k)}_e - \av_{\YY^{(k)}_e}
    \YY^{(k)}_e$.
    Then, $\av_{\YY^{(k)}_e} \XX^{(k)}_e = 0$.
    Consider the matrix when normalized by $\LL$, we have
    \[
        \norm{\ol{\XX^{(k)}_e}} 
        \leq
        \max\curlyparen{\norm{\ol{\YY^{(k)}_e}}, 
        \norm{\av_{\YY^{(k)}_e} \ol{\YY^{(k)}_e}}}
        \leq
        \alpha,
    \]
    where the first inequality is due to the PSD of $\YY^{(k)}_e$ and
    subsequently the PSD of $\av_{\YY^{(k)}_e} \ol{\YY^{(k)}_e}$ and the second
    inequality is due to \autoref{lemma:twalksnew}. 
    Thus, we can set the norm bound $R = \alpha$ in \autoref{thm:freedman}. 

    On to the predicatable quadratic variation. 
    Let 
    \[
        \WW^{(k)} =
        \sum_{i=1}^{k} \sum_{e} \av_{\YY^{(i)}_{e}}
        \sqparen{\ol{\XX^{(i)}_{e}}^2}.
    \]
    Note that since any $\ol{\XX^{(i)}_{e}}^2$ is PSD, any intermediate
    $\WW^{(k)}_e \pleq \WW^{(k)}$ and any $\WW^{(j)}$ is also
    subsumed by $\WW^{(k)}$ for $j < k$. 
    Thus, it suffices to only consider $\WW^{(t)}$. 
    For any $k$ and $e$, 
    \[
        \av_{\YY_e^{(k)}} \ol{\XX_e^{(k)}}^2 
        =
        \av_{\YY_e^{(k)}} \ol{\YY_e^{(k)}}^2 
        -
        (\av_{\YY_e^{(k)}} \ol{\YY_e^{(k)}})^2 
        \pleq 
        \av_{\YY_e^{(k)}} \ol{\YY_e^{(k)}}^2.
    \]
    Then, by \multiref{lemma:unbiased,lemma:twalksnew}, 
    \begin{align*}
        \sum_e \av_{\YY_e^{(k)}} \ol{\XX_e^{(k)}}^2 
        &\pleq 
        \sum_e \av_{\YY_e^{(k)}} \ol{\YY_e^{(k)}}^2 
        \pleq
        \sum_e \av_{\YY_e^{(k)}} \norm{\ol{\YY_e^{(k)}}} \ol{\YY_e^{(k)}}
        \\
        &\pleq 
        \alpha \ol{\Sc(\LL_{G^{(k-1)}}, C)}
        \pleq 
        \alpha \ol{\LL_{G^{(k-1)}}}
        \pleq 
        \alpha \ol{\LL^{(k-1)}},
    \end{align*}
    where the last two inequalities are immediate from the respective block
    Cholesky factorizations of $\LL_{G^{(k-1)}}$ and $\LL^{(k-1)}$ and
    \autoref{fact:twoside}. 
    Now, by our assumption
    \[
        \WW^{(t)} 
        \pleq 
        \alpha \sum_{i=0}^{t-1} \ol{\LL^{(i)}}
        \pleq 
        1.3(t-1)\alpha \II
    \]
    and that we can take 
    $\sigma^2 = C \alpha \log n$ due to $t \leq d = O(\log n)$ 
    by \eqref{eq:dbound}.

    By applying the matrix Freedman inequality (\autoref{thm:freedman}), the
    upperbound of probability of failure satisfies
    \[
        n \cdot \exp\paren{-\frac{0.3^2/2}{C\alpha\log n + 0.3\alpha/3}} 
        \leq
        \frac{1}{n^\delta}
    \]
    when $\alpha^{-1} = \Theta(\log^2 n)$ given that $\delta$ is constant, 
    as required. 
    Note that $0.7\LL \pleq \LL^{(t)} \pleq 1.3 \LL$ is equivalent to 
    $\norm{\ol{\LL^{(t)}} - \PPi} \leq 0.3$ and that 
    $\LL^{(t)} - \LL = \sum_{k=1}^t \sum_e \XX_e^{(k)}$. 
    Combine with the probability of success of at least $1-\frac{t-1}{n^\delta}$
    up to $t-1$, the probability of success up to $t$ is at least 
    \[
        (1-\frac{t-1}{n^\delta})(1 - \frac{1}{n^\delta})
        \geq 
        1 - \frac{t}{n^\delta},
    \]
    again, as required. 

    Lastly, let us analyze the runtime complexity. 
    We seperate the runtime into three parts: (a) the total runtime of applying
    \DS, (b) the total runtime of performing
    \TW, and (c) additional one-time overhead. 
    Recall for (c) that the only overhead is from converting multi-graph into a
    simple graph for $G^{(d)}$ which runs in 
    $O(m)$ work and $O(\log m)$ depth.
    
    Consider (a). 
    By \autoref{lemma:sdd}, each iteration of \DS runs in 
    $O(m)$ work and $O(\log m)$ depth in expectation. 
    As this guarantee merely depends on the parameter of the input
    multi-graph, the work and depth bounds for all iterations are independent.
    Now, using Chernoff bound, we get a total work and depth of $O(m\log n)$
    and $O(\log m \log n)$ using $d = O(\log n)$.

    \sloppy
    As for (b), since \autoref{lemma:twtime} has the guarantee written as
    with high probability, the total work and depth are
    $O(m \log n)$ and $O(\log m \log n)$ again by $d = O(\log n)$. 

    Therefore, combining (a), (b) and (c), we have the total runtime of 
    $O(m \log n)$ work and $O(\log m \log n)$ depth. 
\end{proof}

\section{Achieving $\alpha$-boundedness Initially} \label{sec:alpha}
We sketch out a brief justification to \autoref{lemma:sparsebound} in this
section.
Since the complete justification is rather involved and does not deviate from
\cite{CohenLMMPS15,kyng2017approximate}, we only emphasize on the parallel work
and depth.
The high-level algorithm in \autoref{lemma:sparsebound} runs as follows:
\begin{enumerate}
    \item Sample a sparser graph $G'$ with $\frac{m}{K}$ edges by
        uniformly choosing edges and scaling weights appropriately.
    \item Compute leverage score overestimations of $G$ using $G'$.
    \item Split edges into $\alpha$-bounded multi-edges using the leverage score
        overestimations.
\end{enumerate}
We remark that (3) is different compared to \cite{CohenLMMPS15,
kyng2017approximate} which use the leverage score estimations for
sparsification instead.

For the first step, we again use the parallel sampling algorithm from
\cite{HS19} and transform the graph $G'$ into appropriate representation in
parallel.
The reweighting can be done independently. 
Thus, step (1) runs in $O(m)$ work and $O(\log n)$ depth.

The second step requires the standard dimension reduction approach due to
\cite{SS08,KoutisLP15} using the Johnson–Lindenstrauss Lemma for querying
leverage score overestimations.
Again, we omit the detailed analysis on the estimation factor and focus on the
parallel runtime.
There are three major steps that we need to consider: 
(a) sampling a uniform random matrix of dimension $O(\log n) \times (m/K)$, 
(b) solving $O(\log n)$ systems of $\LL_{G'}$ to $O(1)$-approximations,
(c) for each edge in $G$, compute the $l_2$ distance of two vectors of size 
$O(\log n)$.
All additional one-time overhead is subsumed by these operations.
For (a), all entries are sampled by independent Bernoulli random variables,
giving us $O(m\log n/K)$ work and $O(1)$ depth.
For (b), we use our solver in \autoref{thm:main} in parallel for each system,
resulting in $O(\frac{m}{K} \log^4 n \log\log n)$ work and 
$O(\log^2 n \log\log n)$ depth.
As for (c), computing the distances can be done in parallel for each edge,
giving work $O(m\log n)$ and depth $O(\log\log n)$.
In total, we get $O(m\log n + \frac{m}{K} \log^4 n \log\log n)$ work and 
$O(\log^2 n \log\log n)$ depth.

Note that after (2), the leverage score overestimations $\widehat{\ttau}$
satisfy that $\sum_{e} \widehat{\ttau}(e) \leq O(nK)$. 
In step (3), it suffices to split an edge $e$ into
$\ceil{\alpha^{-1}\widehat{\ttau}(e)}$ copies, resulting in $O(m+nK\alpha^{-1})$
multi-edges and a parallel runtime of $O(m+nK\alpha^{-1})$ work and $O(\log n)$
depth.

To summarize, we get a total parallel runtime of 
$O(nK\alpha^{-1} + m\log n + \frac{m}{K}\log^4 \log\log n)$ work and 
$O(\log^2 n \log\log n)$ depth as desired.

\section{Schur Complement Approximation} 
\label{sec:scappr}

\begin{algorithm}[t]
    \caption{Schur Complement Approximation}
    \label{alg:scappr}
    \Pro{\AS{$G, C$}}{
        Initially set $G^{(0)} \gets G$, $k \gets 0$ and 
        $U_0 \gets V \backslash C$.\;
        \While{$U_{k} \neq \emptyset$}{
            $k \gets k+1$.\;
            $F_k \gets$ \DS{$G^{(k-1)}[U_{k-1}]$} and set $C_k
            \gets C_{k-1} \backslash F_k, U_k \gets U_{k-1} \backslash F_k$.\;
            $G^{(k)} \gets$ \TW{$G^{(k-1)}, C_k$}.\;
        }
        \Return{$G_S = G^{(d)}$ where $d$ is the last $k$.}\;
    }
\end{algorithm}

In this section, we show that with a slight modification to our main algorithm
\BC, we can compute a sparse approximation to the Schur
complement of $\Sc(\LL,C)$ in parallel. 
Our algorithm \AS is similar to that given by \cite{DKP+17}. 
\begin{theorem} \label{thm:scappr}
    There exists some $\alpha^{-1} = \Theta(\epsilon^{-2}\log^2 n)$ where $0 <
    \epsilon < 1/2$
    such that
    for any multi-graph $G$ with $n$ vertices and
    $m$ multi-edges that are all $\alpha$-bounded w.r.t.~$\LL_G$
    and any non-trivial $C \subset V$, 
    \AS{$G,C$} returns a multi-graph $G_S$ such that 
    with high probability the following holds: 
    \begin{enumerate}
        \item $\LL_{G_S} \approx_\epsilon \Sc(\LL_G, C)$.
        \item The number of multi-edges in $G_S$ is at most $m$.
    \end{enumerate}
    If $s = |V \backslash C|$, then with high probability the algorithm runs in
    $O(m \log s)$ work and $O(\log s \log m)$ depth.
\end{theorem}
\begin{proof}
    This proof is analogous to that of \autoref{thm:bcfapproxnew}. 

    Note that a 5DD subset of an induced subgraph is a 5DD subset of the graph
    itself.
    By \autoref{lemma:sdd}, the total number of iterations of \AS is bounded by
    $d \leq \log_{40/39}(|U_0|) = O(\log s)$,
    where $s = |U_0| = |V \backslash C| \leq n$, which is also in $O(\log n)$.
    Note that a subset of size 1 is always 5-DD.
    So, even when $n \leq 40$, the guarantees in \autoref{lemma:sdd} still holds
    true.

    The approximate Schur complement is simply the bottom right block of
    $\DD^{(d)}$ by \eqref{eq:blockd}. 
    Thus, the desired approximation guarantee follows direcly from
    the approximation guarantee for the factorization 
    $(\UU^{(d)})^\top \DD^{(d)} \UU^{(d)}$.
    We can then follow the same argument of \autoref{thm:bcfapproxnew} using
    \autoref{thm:freedman}.

    Again, each intermediate Schur complement has at most $m$ multi-edges and
    that the number of multi-edges in $G_S$ is at most $m$ as well.
    Moreover, combine with $d = O(\log s)$, the algorithm runs in $O(m \log s)$
    work and $O(\log s \log m)$ depth.
\end{proof}

\printbibliography

\appendix

\section{Missing Proofs}
\label{sec:proof}

\begin{proof}[Proof of \autoref{lemma:sdd}, \cite{LPS15,KyngLPSS16}]
    As $F$ is a subset of $F'$, 
    \[
        \sum_{e \in E(G[F]) : e \ni i} \ww(e)
        \leq 
        \sum_{e \in E(G[F']) : e \ni i} \ww(e).
    \]
    So it is guaranteed that returned $F$ is $5$-DD. 

    Consider the runtime of \textsc{5DDSubset}. 
    As each iteration of the algorithm runs in $O(m)$ work and $O(\log m)$
    depth, it suffices to show that the probability of the algorithm finishs in
    each iteration is some constant, specifically $1/2$. 

    Let $A_i$ be the event that $i \in F'$ and $i \notin F$. 
    The set $F$ is then the set of $i \in F'$ for which $A_i$ does not hold. 
    Notice that $A_i$ only happens if $i \in F'$ and 
    \[
        \sum_{e \in E(G[F']) : e \ni i} \ww(e) > \frac{1}{5} \ww(i),
    \]
    where $\ww(i) = \sum_{e \in E: e \ni i} \ww(e)$. 
    Given that $i \in F'$, the probability that $j \in F'$ for $j \neq i$ is
    \[
        \frac{1}{n-1} \paren{\frac{n}{20} - 1}.
    \]
    Then,
    \[
        \av \left(  \sum_{e \in E(G[F']) : e \ni i} |\LL_{ij}| \middle| i
          \in F' \right)
        \leq%
        \frac{1}{n-1}\paren{\frac{n}{20} - 1} \sum_{e \in E : e \ni i} \ww(e)
        <
        \frac{1}{20} \sum_{e \in E : e \ni i} \ww(e)
        \leq
        \frac{1}{20} \ww(i).
    \]
    By Markov's inequality, 
    \[
        \pr \sqparen{\sum_{e \in E(G[F']) : e \ni i} \ww(e) > 
        \frac{1}{5} \ww(i) \big| i \in F'} 
        < 1/4,
    \]
    and that 
    \[
        \pr(A_i) = \pr(i \in F')\pr(i \notin F | i \in F') <
        \frac{1}{20}\frac{1}{4} = \frac{1}{80}.
    \]
    Using Markov's inequality for a second time gives us
    \[
        \pr\paren{|\{i : A_i\}| \geq \frac{n}{40}} < 1/2
    \]
    as required. 
\end{proof}

\begin{proof}[Proof of \autoref{lemma:sumofwalks}]
    Consider the Schur complement 
    $\Sc(\LL,C) = \LL_{CC} - \LL_{CF}\LL_{FF}^{-1}\LL_{FC}$. 
    Let us write $\LL_{FF} = \DD - \AA$ where $\DD$ is diagonal and $\AA$ is
    nonnegative with 0 diagonally. 
    The matrix $\LL_{FF}$ must be diagonally dominant, and thus,
    $-\II \ple \DD^{-1/2}\AA\DD^{-1/2} \ple \II$. 
    Using the identity that
    \[
        (\II - \MM)^{-1} = \sum_{i=0}^\infty \MM^i
    \]
    for $\norm{\MM} < 1$, we have, 
    \begin{align*}
        \LL_{FF}^{-1} 
        &= 
        (\DD - \AA)^{-1}
        = 
        \DD^{-1/2} \paren{\II - \DD^{-1/2}\AA\DD^{-1/2}}^{-1}\DD^{-1/2}
        \\
        &= 
        \DD^{-1/2} \sqparen{\sum_{i=0}^\infty \paren{\DD^{-1/2}\AA\DD^{-1/2}}^k}
        \DD^{-1/2}
        =
        \sum_{i=0}^\infty (\DD^{-1}\AA)^k\DD^{-1}.
    \end{align*}
    Substituting this in place of $\LL^{-1}_{FF}$ gives
    \[
        \Sc(\LL,C) = \LL_{CC} - 
        \sum_{i=0}^\infty \LL_{CF} (\DD^{-1}\AA)^i \LL_{FC}.
    \]
    As $\LL_{FC},\LL_{CF}$ are non-positive, we can replace them with
    $-\LL_{FC}, -\LL_{CF}$ respectively to makes the terms positive. 
    Recall that each entry of $\AA$ is defined as the weighted sum of all
    multi-edges connecting them $\AA_{uv} = \sum_{e : e \ni u,v}\ww(e)$. 
    Now, fix two distinct vertices $s,t \in C$. 
    For any fixed $k \geq 1$, 
    \[
        \paren{\LL_{CF}(\DD^{-1}\AA)^{k-1}\DD^{-1}\LL_{FC}}_{s,t}
        =
        \sum_{
            \begin{subarray}{c}
                W: C-\text{terminal} \\
                |E(W)| = k+1; s,t \in V(W)
            \end{subarray}
        }
        \frac{\prod_{i=1}^{k+1} \ww(e_i)}{\prod_{i=1}^{k-1} \ww(u_i)}.
    \]
    Summing over all $k$ as long with the weights of multi-edges of $(s,t)$ in
    $\LL_{CC}$ gives the correct identity. 
\end{proof}

\begin{proof}[Proof of \autoref{lemma:jacobi}]
    Recall the definition of $\ZZ$ from \eqref{eq:jacobi}. 
    \[
        \ZZ = \sum_{i=0}^l \XX^{-1} (-\YY\XX^{-1})^i,
    \]
    where $l$ is an odd integer such that $l \geq \log_2 3/\epsilon$. 

    The left-hand inequality is equivalent to the statement that all the
    eigenvalues of $\ZZ\MM$ are at most 1 (\cite{SpielmanT04}). 
    To see that this is the case,
    \begin{align*}
        \ZZ\MM 
        &= 
        \sum_{i=0}^l (-\XX^{-1}\YY)^i \XX^{-1} (\XX + \YY)
        \\
        &=
        \sum_{i=0}^l (-\XX^{-1}\YY)^i - \sum_{i=1}^{l+1}(-\XX^{-1}\YY)^i
        =
        \II - (\XX^{-1}\YY)^{l+1},
    \end{align*}
    where the last equality uses the assumption that $l$ is odd. 
    Since $l+1$ is even, all eigenvalues of $(\XX^{-1}\YY)^{l+1}$ are
    nonnegative. 
    Then, by the last matrix, all eigenvalues of $\ZZ\MM$ are at most 1 as
    required. 

    The RHS inequality is equivalent to all eigenvalues of $\ZZ(\MM +
    \epsilon\YY)$ are at least 1.
    Expanding this gives
    \[
        \sum_{i=0}^l (-\XX^{-1}\YY)^i \XX^{-1} (\XX + (1+\epsilon)\YY)
        =
        \II - (\XX^{-1}\YY)^{l+1} - \epsilon \sum_{i=1}^{l+1}(-\XX^{-1}\YY)^i.
    \]
    Then, the eigenvalues of this matrix ar of the form
    \[
        1 - \lambda^{k+1} - \epsilon \sum_{i=1}^{l+1} (-\lambda)^i,
    \]
    where $\lambda$ ranges over the eigenvalues of $\XX^{-1}\YY$.  
    Given that $\MM$ is $5$-DD, we have $2\YY \pleq \XX$. 
    Then, the eigenvalues of $\XX^{-1}\YY$ must be in the range of $[0,1/2]$. 
    Once again, by $l$ being odd, it suffices to have
    \[
        \epsilon - \lambda^{l+1} - \epsilon \frac{1+\lambda^{l+2}}{1+\lambda}
        \geq 0.
    \]
    The LHS is minimized at $\lambda = 1/2$, which gives us that 
    $\lambda \geq \log_2(3/\epsilon+1)-1$. 
    As $\epsilon < 1$, it suffices to instead have 
    $\lambda \geq \log_2(3/\epsilon)$. 
\end{proof}

\begin{proof}[Proof of \autoref{lemma:jacobiapprox}]
    Let the matrix in \autoref{lemma:jacobiapprox} be $\Ltil$. 
    We start by defining a matrix
    \[
        \Lhat =
        \begin{pmatrix}
            \ZZ^{-1} & \LL_{FC} \\
            \LL_{CF} & \LL_{CC}
        \end{pmatrix}.
    \]
    By \autoref{lemma:jacobi} and \autoref{fact:block}, 
    \[
        \LL \pleq \Lhat \pleq 
        \LL + \epsilon 
        \begin{pmatrix} 
            \YY & 0 \\
            0 & 0
        \end{pmatrix},
    \]
    where $\YY$ is Laplacian with off-diagonal entries the same as $\LL_{FF}$. 
    This means that $\YY$ is the Laplacian of a induced subgraph of $\LL$ and
    that
    \[
        \begin{pmatrix} 
            \YY & 0 \\
            0 & 0
        \end{pmatrix}
        \pleq \LL.
    \]
    Then, $\LL \pleq \Lhat \pleq (1+\epsilon)\LL$.
    Now, by \autoref{fact:inver} and \eqref{eq:linv}, the inverse schur
    complement onto $C$ is a submatrix of the Laplacian and that 
    $e^{-\epsilon} \Sc^+(\LL,C) \pleq \Sc^+(\Lhat,C) \pleq \Sc^+(\LL,C)$
    since $e^{-\epsilon} < 1/(1+\epsilon)$.
    Using \autoref{fact:inver} again gives us
    \[
        e^{-\epsilon} \Sc(\Lhat,C) \pleq \Sc(\LL,C) \pleq \Sc(\Lhat,C).
    \]
    Then, adding $\ZZ$ onto the $FF$ block gives
    \[
        e^{-\epsilon}
        \begin{pmatrix}
            \ZZ^{-1} & 0 \\
            0 & \Sc(\Lhat,C)
        \end{pmatrix}
        \pleq
        \begin{pmatrix}
            \ZZ^{-1} & 0 \\
            0 & \Sc(\LL,C)
        \end{pmatrix}
        \pleq 
        \begin{pmatrix}
            \ZZ^{-1} & 0 \\
            0 & \Sc(\Lhat,C)
        \end{pmatrix}.
    \]
    By \autoref{fact:twoside} with $C$ set to the lower triangular matrix,
    $e^{-\epsilon} \Lhat \pleq \Ltil \pleq \Lhat$.
    Combine with the approximation on $\Lhat$ and by the fact 
    $1+\epsilon < e^\epsilon$, 
    \[
        e^{-\epsilon} \LL \pleq \Ltil \pleq e^{\epsilon} \LL,
    \]
    which concludes our proof. 
\end{proof}

\end{document}